\documentclass[letterpaper, 10 pt, conference]{ieeeconf}  
\usepackage{cite}
\usepackage{paralist}
\usepackage{amsfonts}
\usepackage{graphics} 
\usepackage{epsfig} 
\usepackage{mathptmx} 
\usepackage{times} 
\usepackage{amsmath} 
\usepackage{amssymb}  
\usepackage{physics}
\usepackage{mdwmath}
\usepackage{mdwtab}
\usepackage{color}
\usepackage[hidelinks]{hyperref}
\usepackage{algpseudocode}
\usepackage{algorithm}
\usepackage{caption}
\usepackage{subcaption}
\usepackage{comment}

\newtheorem{thm}{Theorem}
\newtheorem{prop}{Proposition}
\newtheorem{remark}{Remark}
\newtheorem{definition}{Definition}

\newcommand{\Epsilon}{\mathrm{E}}

\newcommand{\Proj}{\text{Proj}}

\usepackage{bm}
\makeatletter

\title{\LARGE \bf
State Estimation and Control for Stochastic Quantum Dynamics with Homodyne Measurement: Stabilizing Qubits under Uncertainty
}

\author{Nahid Binandeh Dehaghani, A. Pedro Aguiar, Rafal Wisniewski
\thanks{N. B. Dehaghani and A. P. Aguiar are with the Research Center for Systems and Technologies (SYSTEC), ARISE, and the Electrical and Computer Engineering Department, Faculty of Engineering, University of Porto, Rua Dr. Roberto Frias, 4200-465 Porto, Portugal
{\tt\small \{nahid,pedro.aguiar\}@fe.up.pt}}
\thanks{R. Wisniewski is with Department of Electronic Systems, Aalborg University, Fredrik Bajers vej 7c, DK-9220 Aalborg, Denmark
        {\tt\small raf@es.aau.dk}}
\thanks{
The authors gratefully acknowledge the financial support provided by the Foundation for Science and Technology (FCT/MCTES) within the scope of the PhD grant 2021.07608.BD, the Associated Laboratory ARISE (LA/P/0112/2020), the R\&D Unit SYSTEC through Base (UIDB/00147/2020) and Programmatic (UIDP/00147/2020) funds and project RELIABLE (PTDC/EEI-AUT/3522/2020), all supported by national funds through FCT/MCTES (PIDDAC). The work has been done in honor and memory of Professor Fernando Lobo Pereira.}}

\begin{document}

\maketitle
\thispagestyle{empty}
\pagestyle{empty}

\begin{abstract}
This paper introduces a Lyapunov-based control approach with homodyne measurement. We study two filtering approaches: (i) the traditional quantum filtering and (ii) a modified version of the extended Kalman filtering. We examine both methods in order to directly estimate the evolution of the coherence vector elements, 
using sequential homodyne current measurements. The latter case explicitly addresses the dynamics of a stochastic master equation with correlated noise, transformed into a state-space representation, ensuring by construction the quantum properties of the estimated state variable. 
In addition, we consider the case where the quantum-mechanical Hamiltonian is unknown, and the system experiences uncertainties. 
In this case, we show as expected that both filters lose performance, exhibiting large expected estimation errors. To address this problem, we propose a simple multiple model estimation scheme that can be directly applied to any of the studied filters.
We then reconstruct the estimated density operator \( \hat{\rho} \), describing the full state of the system, and subject it to a control scheme. The proposed switching-based Lyapunov control scheme, which is fed with \( \hat{\rho} \), guarantees noise-to-state practically stable in probability of the desired stationary target set with respect to the estimation error variance. We demonstrate our approach's efficacy in stabilizing a qubit coupled to a leaky cavity under homodyne detection in the presence of uncertainty in resonance frequency.
\end{abstract}


\maketitle

\section{Introduction}
\label{sec:introduction}
\PARstart{A}{} particularly fertile domain within quantum control is that of feedback control, where the control strategy is dynamically adjusted based on real-time feedback signals emanating from the evolving system \cite{sayrin2011real}. In this regard,
significant advancements have been achieved in the development and enhancement of the underlying principles of quantum probability and filtering theory \cite{belavkin1983theory, belavkin1992quantum, bouten2007introduction, gough2012quantum, rouchon2015efficient,emzir2017quantum}.

\begin{figure}[t]
	\centering
	\includegraphics[scale=0.3]{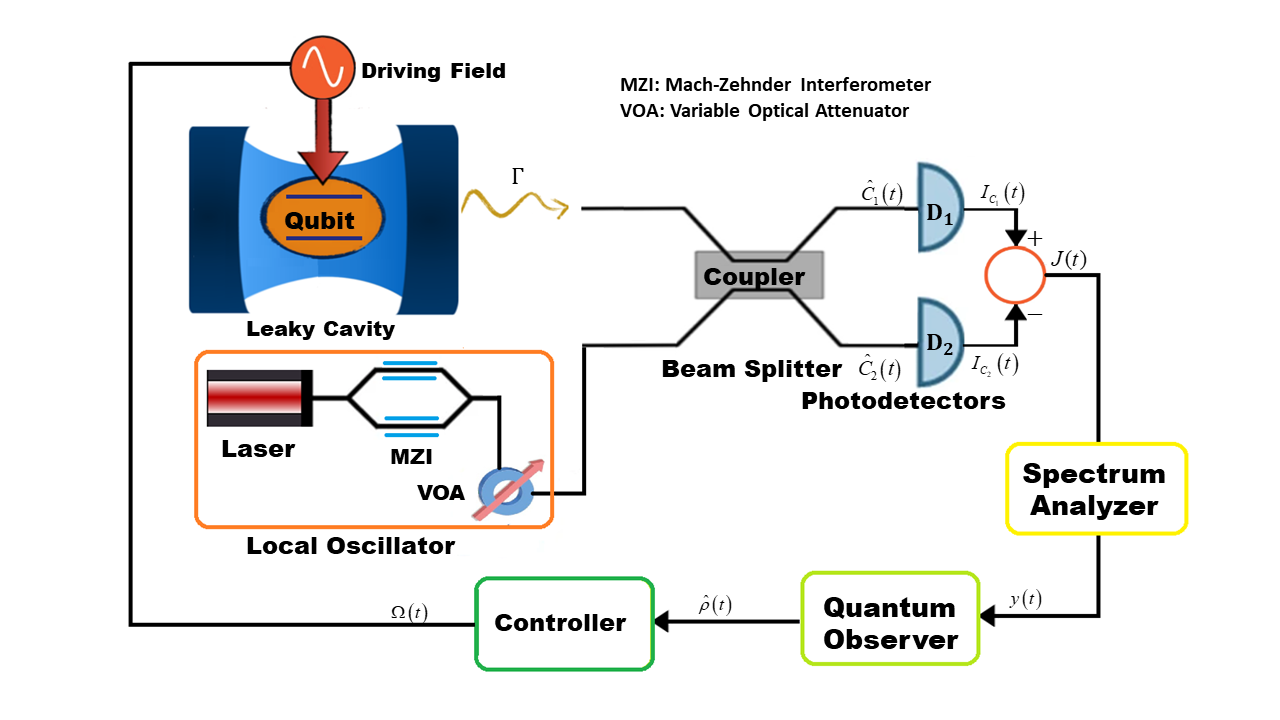}
        \caption{\small Illustration of the proposed configuration: A qubit coupled to a leaky cavity is continuously monitored via homodyne detection. The emitted radiation from the qubit is combined with a local oscillator laser at a beam splitter. The transformations of the field operators through the beam splitter are denoted by $C_1(t)$ and $C_2(t)$, which are then detected by photodetectors \( D_1 \) and \( D_2 \), respectively. The difference of the two produced signals ${I_c}_1$ and ${I_c}_2$ generates the resulting signal $J(t)$. This signal is sent to a quantum observer, which estimates the quantum state \( \hat{\rho}(t) \). The estimated state is then transmitted to a controller. The controller adjusts the driving field $\Omega(t)$ through the control Hamiltonian to the qubit, aiming to generate and stabilize the qubit into a desired target set.} 
        \label{setup}
\end{figure}
In this study, see Fig. \ref{setup}, we spotlight a qubit undergoing continuous homodyne detection described by a stochastic master equation. Creating a controller is crucial because it it customizes the control system to manage specific state transformation tasks based on the homodyne current measurements. 
Inherent stochasticity of the analyzed problem accentuates the complexity of the control task, as the controller must dynamically adapt to the random evolution of the quantum state in each trajectory. Additionally, discerning the actual state of the qubit solely from the instantaneous homodyne current proves challenging. Thus, a sophisticated estimating mechanism is imperative to extract pertinent information about the qubit's state from the time series of measurement results. Herein, we outline our contributions.

\begin{itemize}

\item We elaborate on three key estimation techniques: Quantum Filter Estimator, Quantum Kalman Filter Estimator, and Multiple-Model Adaptive Estimator. By leveraging an equivalent state space representation of a general mixed state of a qubit, we offer a robust mathematical framework. This framework describes the qubit's state evolution through stochastic system dynamics in Itô form, which can conveniently be used in the estimation techniques. 
In this framework, we first implement the conventional quantum filtering method \cite{bouten2007introduction} as a fundamental tool for quantum state estimation, however for the stochastic system dynamics in Itô state space. Secondly, the extension of the classical Kalman filter to the quantum domain is addressed by handling the nonlinearity of the stochastic system dynamics and the correlation between process and measurement noise, utilizing a decorrelating approach. 
Additionally, we integrate the structural properties of a valid quantum state into the observer equations, ensuring that state estimates always belong to the Bloch ball, and propose a projection operator to maintain the physical relevance of estimated states. The construction of a quantum Kalman filter estimator is detailed, continuously updating state estimates as new measurements are acquired. It is worth noting that Kalman-based filtering methodologies have been previously implemented in the quantum domain \cite{emzir2017quantum} and have found applications in quantum systems \cite{ma2022review}. Here, we go beyond the usual Kalman techniques, e.g., in \cite{ corcione2023state}, by explicitly addressing the fact that the quantum system exhibits correlated process noise with measurement noise. More importantly, in our design, we guarantee by construction the properties of valid quantum states. 
Thirdly, we develop a class of Multiple-Model Adaptive Estimators (MMAEs) to address the challenge of parametric model uncertainty when the quantum-mechanical Hamiltonian is unknown or subject to fluctuations. As anticipated, both filters suffer from reduced performance under these conditions, leading to significant estimation errors. To mitigate this, we employ parallel filters, each tailored to a specific model, combine their outputs by dynamically adjust weights based on measurement data. We show that the MMAE exhibits good performance and can be directly applied to any of the proposed Quantum Filter and Quantum Kalman Filter estimators, demonstrating their applicability and effectiveness in different scenarios. 
These contributions collectively advance the state of the art in quantum state estimation, providing methods and frameworks for effective quantum system management.

\item We introduce a robust switching-based Lyapunov control strategy tailored explicitly for 
quantum stochastic systems with state uncertainty. It is important to recognize that Lyapunov control approaches have been widely used in the field of quantum control for ensuring the stability of quantum systems or guiding them to a desired state \cite{mirrahimi2007stabilizing, hou2012optimal, cardona2018exponential}. More recently, a quantum measurement-based feedback control via switching strategies has been proposed in \cite{grigoletto2021stabilization}
for quantum systems that can be controlled with a series of driving dynamics. In addition, \cite{liang2024dissipative}
analyzes switching techniques designed for the rapid stabilization of pure states and subspaces within quantum filtering equations.
Our contribution lies in the design of a controller in a different setup of \cite{grigoletto2021stabilization, liang2024dissipative} with the aim of steering the system state $\rho$ towards a pre-defined stationary target set, by directly using the estimate of $\rho$. In addition, we explicitly show that the closed-loop system guarantees noise-to-state practically stable in probability with respect to the variance of the estimation error. This means that the distance of $\rho(t)$ to the target set is bounded in expectation by some monotone function of the $\sup_{\tau \in [0, t]} \tilde\sigma^2(t)$, where $\tilde\sigma^2(t)$ is the incremental variance associated to the noisy estimation error $\Tr{\tilde{\rho}}=\Tr{\rho-\hat{\rho}}$.


\item We consider a complete framework that combines quantum filtering with switching Lyapunov-based feedback control, ensuring robustness despite uncertainty and noise. It probabilistically guarantees practical convergence of the system's state to a small neighborhood of the desired invariant set, whose size degradates gracefully with the intensity of the estimation error.
\end{itemize}

\subsection{Notation}
We use the superscripts \( T \) and \( \dagger \) to denote the transpose and conjugate transpose of a matrix, respectively. Partial derivatives are denoted using \( \partial \), where \( \frac{\partial f}{\partial t} \) signifies the partial derivative with respect to the variable \( t \). The notation \([ \cdot, \cdot ]\) represents a commutator of matrices. The imaginary unit is denoted by \( i = \sqrt{-1} \), and the identity matrix is denoted by \( I \). For a generic matrix \( A \), the trace is shown by the operator \( \Tr{A} \). 
The set of all quantum density operators $\rho$ in a 
finite-dimensional complex Hilbert space $\mathbb{H}$ is denoted by ${\mathbb{H}}_\rho = \left\{ \rho \in \mathcal{L}(\mathbb{H}) \ \middle|\ \rho = \rho^\dagger, \ \rho \geq 0, \ \text{Tr}(\rho) = 1 \right\}$, where $\mathcal{L}(\mathbb{H})$ denotes the set of linear operators on $\mathbb{H}$. For a variable \( X \), the classical expectation value is indicated by \( \mathbb{E}[X] \), while the quantum expectation value is denoted by \( \langle X \rangle = \Tr{X \rho} \). The notation \( U(2) \) represents the unitary group consisting of all \(2\times 2 \) unitary matrices, which are relevant to qubit operations. 
We denote a continuous function \(\zeta: \mathbb{R}_+ \rightarrow \mathbb{R}_+\) to be of class \(\mathcal{K}\) if it satisfies the following two properties: (i) Positive definiteness: \(\zeta(s) > 0\) for all \(s \in \mathbb{R}_+ \setminus \{0\}\), and \(\zeta(0) = 0\), and (ii)
strict monotonicity: \(\zeta\) is strictly increasing. We denote a continuous function \(\eta: \mathbb{R}_+ \times \mathbb{R}_+ \rightarrow \mathbb{R}_+\) to be of class \(\mathcal{K}\mathcal{L}\) if it satisfies the following two properties: (i) For each fixed \(t\), the function \(\eta(\cdot, t)\) is of class \(\mathcal{K}\), and (ii) For each fixed \(s\), the function \(\eta(s, \cdot)\) is decreasing and \(\lim_{t \to \infty} \eta(s, t) = 0\).


\section{Stochastic Quantum Control and Estimation}
Stochastic quantum dynamics encompass the complex interactions between quantum systems and classical control mechanisms. This section delves into the mathematical frameworks and problem statements essential for controlling and estimating the states of such quantum systems. 
\subsection{QUANTUM STOCHASTIC EVOLUTION}
The stochastic quantum dynamics can be accurately described by well-structured stochastic differential equations known as Stochastic Master Equations (SME). These equations define the relationships between the classical control input $\Omega(t)$, tuned by the controller, and the classical output $J(t)$ corresponding to instantaneous value of the observed measurements. For such a setup, as indicated in Fig. \ref{setup}, the diffusive quantum stochastic master equation, interpreted in the context of Itô calculus, is expressed as \cite{rouchon2022tutorial}
\begin{align}\label{smeftp}
d\rho \left( t \right) &=  -i\left[ {H_d+\Omega(t)H_c,\rho \left( t \right)} \right]dt + \Gamma\mathcal{D}\left[ {{c_d}} \right]\rho \left( t \right)dt \notag \\
 &\!\!\!\!\!\!\!\!\!\quad +M\mathcal{D}\left[ {{c_m}} \right]\rho \left( t \right)dt + \sqrt {\eta M}\mathcal{H}\left[ {{c_m}} \right]\rho \left( t \right)dW\left( t \right),
\end{align}
where the quantum density operator $\rho \in {\mathbb{H}}_\rho $.
The unitary aspects of the dynamics \eqref{smeftp} is characterized by the quantum mechanical system Hamiltonian $H(t)$, composed of drift ${H}_d$ and control ${H}_c$ terms. The dissipative part of the dynamics is described  
through arbitrary Lindblad operators $c_d$ and $c_m$, corresponding to decoherence and measurement channels. The superoperators $\mathcal{D}\left[ {{c}} \right]\rho$ corresponding to dissipation, and $\mathcal{H}\left[ {{c}} \right]\rho$ corresponding to measurement for a generic operator $c$ are defined as follows
\begin{equation*}
    \begin{aligned}
        &\mathcal{D}[c]\rho=c\rho c^\dagger-\frac{1}{2}(c^\dagger c\rho+\rho c^\dagger c),\\
        &\mathcal{H}[c]\rho=c{{\rho }}+{{\rho }}{{{c}}^{\dagger }}-{{\left\langle c+{{{c}}^{\dagger }} \right\rangle }}{{\rho }}.
    \end{aligned}
\end{equation*}
The coefficient $\Gamma>0$ indicates the decay rate, $M>0$ is the interaction strength corresponding to the measurement process, and $\eta \in [0, 1]$ indicates the detector efficiency with $\eta=1$ implying perfect detection and $\eta=0$ lack of information about the quantum state. The differential $dW(t)$ denotes an infinitesimal Wiener increment satisfying Itô's lemma with $\mathbb{E}[dW(t)]=0$, $\mathbb{E}[dW^2(t)]=dt$. The same Wiener process is shared by the outcome of classical measurement homodyne current
\begin{equation}\label{J}
J(t) = \sqrt{\eta M}{{\left\langle c_m+{{{c_m^{\dagger }}}} \right\rangle }}+\xi(t),
\end{equation}
where $\xi(t) = dW(t)/dt$, follows Gaussian white noise properties with $ \mathbb{E}[\xi(t)\xi(t')] = \delta(t-t')$.

\subsection{Problem Statement}
This subsection outlines the primary research objectives pursued in this paper. The first objective is to develop an observer capable of estimating the evolution of the density operator $\rho(t)$ from a sequence of homodyne current measurements. The second objective is to design a controller that steers the system's state $\rho(t)$ towards a predetermined target state within a specified orbit. These objectives aim to advance the control of stochastic quantum dynamics, providing robust solutions to challenges in quantum state estimation and stabilization.

\textbf{Problem 1:} Given the quantum system presented in \eqref{smeftp}-\eqref{J}, the first problem addressed in the paper is to develop an observer to directly estimate the evolution of the density operator $\rho(t)$ from the continuous sequence of homodyne current measurements. 
 In addition, we consider that the drift Hamiltonian $H_d$ contains \emph{unknown} constant parameters. \\

Prior to expressing the second problem and delving into the control design, we introduce the following stability definitions pertinent to the stochastic quantum system problem at hand, drawing upon the frameworks in \cite{ito2015stability, deng2001stabilization}.
\begin{definition}\label{def:ApS} (\textit{Asymptotically Practically Stable in Probability}): 
Let $\chi$ be a closed and invariant set under the dynamics governed by the stochastic differential equation in Itô form
\[
d\rho(t) = f(\rho)dt + g(\rho)\sigma(t)dW(t),
\]
where $\rho \in {\mathbb{H}}_\rho $, $\sigma(t)$ is the intensity of noise at time $t$, and $W(t)$ denotes 
a 1-dimensional standard Wiener process. The set $\chi$ is \textit{asymptotically practically stable in probability} if, for every $\varepsilon > 0$, there exists a function $\beta \in \mathcal{K}\mathcal{L}$ and a non-negative real number $b \geq 0$ such that for every initial state $\rho(0) \in \mathbb{H}_\rho$, the following condition holds:
\[
\mathbb{P}\left\{ d(\rho(t), \chi) < \beta\left(d(\rho(0), \chi), t\right) + b \right\} \geq 1 - \varepsilon, \quad \forall t \geq 0
\]
where $d(\rho, \chi)$ denotes a distance from state $\rho$ to the set $\chi$, measured in an appropriate metric. If $b = 0$, $\chi$ is \textit{asymptotically stable in probability}. If $\beta(r, t) = c r e^{-\lambda t}$ for some $c > 0$ and $\lambda > 0$, $\chi$ is \textit{exponentially stable in probability}.
\end{definition}

\begin{definition} \label{def:NSpS} (\textit{Noise-to-State Practically Stable in Probability}):
Adopting the conditions from Definition \ref{def:ApS}, the set $\chi$ is \textit{noise-to-state practically stable} (NSpS) in probability if, for every $\varepsilon > 0$, there exists a function $\beta \in \mathcal{K}\mathcal{L}$, a function $\gamma \in \mathcal{K}$, and a non-negative real number $b \geq 0$, ensuring that for every initial density operator $\rho(0) \in \mathbb{H}_\rho$, the following holds with probability at least $1 - \varepsilon$, for every $t \geq 0$
\begin{equation*}
\mathbb{P}\Big\{ d(\rho(t), \chi) \! <\beta\left(d(\rho(0), \chi), t\right) +
\gamma(\!\!\!\sup_{\tau \in [0, t]} \! \! \sigma^2(\tau))\! \!+ b \Big\} \! \geq \! 1 - \varepsilon 
\end{equation*}
If $b = 0$, $\chi$ is \textit{noise-to-state stable} (NSS) in probability. 
\end{definition}
Definition \ref{def:NSpS} is introduced to measure the resilience of a quantum system's stability in the face of noise. It ensures that, probabilistically, the system's state stays within a practical vicinity of the invariant set $\chi$, notwithstanding the influence of quantum noise. The extent of this vicinity is contingent upon the intensity of the noise, as characterized by the variance
of the 1-dimensional standard Wiener process $W(t)$.

\textbf{Problem 2:} 
The second problem is to design a controller that commands the signal $\Omega(t)$ with the aim of steering the system's state $\rho(t)$ from an arbitrary initial state $\rho_0\in\mathbb{H}_\rho$, towards a stationary point within a predetermined target set $\chi_f$. The set $\chi_f$ is characterized by all the elements $\rho_f\in \mathbb{H}_\rho$ that 
commutes with the drift Hamiltonian ${H}_d$, that is, $[ \rho_f, {H}_d ] = 0$. 
This commutation condition ensures that the target states $\rho_f$ remain invariant under the evolution governed by $H_d$. 



\section{Estimation Techniques in Stochastic Quantum Dynamics}
This section introduces and elaborates on three estimation techniques that can play a crucial role in quantum state estimation: Quantum Filter Estimator, Quantum Kalman Filter Estimator, and  the Multiple-Model Adaptive Estimator to address the uncertainty on the drift Hamiltonian $H_d$. 

In the following, we leverage the fact that 
a general mixed state of a qubit can be expressed as a linear combination of the Pauli matrices, which, along with the identity matrix, form a basis for the $2\times 2$ self-adjoint matrices, that is,
\begin{equation}\label{rcv}
\rho(t) = \frac{1}{2} I + \frac{1}{2} \sum_{k=1}^{3} x_k(t) \sigma_k,
\end{equation}
where the real-valued functions $x_k(t)=\left\langle {{\sigma_k}} \right\rangle$ 
are the coordinates of a point within the unit ball and are taken to be components of a so-called coherence vector
$\mathbf{x}(t) = [x_1(t), x_2(t), x_{3}(t)]^T \in \mathbb{R}^{3}$.
By rewriting \eqref{smeftp} one attains the following stochastic system dynamics for \(\mathbf{x}(t)\) in Itô form
\begin{equation}\label{dx}
d\mathbf{x}(t) = \mathbf{f}(\mathbf{x},\Omega(t)) \, dt + \mathbf{g}(\mathbf{x}) \, dW(t),   
\end{equation}
where  
$\mathbf{f}(\mathbf{x},\Omega)= [f_1(\mathbf{x},\Omega),  f_2(\mathbf{x},\Omega), f_{3}(\mathbf{x},\Omega)]^T$
and 
$\mathbf{g}(\mathbf{x})= [g_1(\mathbf{x}),  g_2(\mathbf{x}), g_{3}(\mathbf{x})]^T$, whose elements for $k=1,2,3$ are
\begin{align*}
\mathbf{f}_k(\mathbf{x}, \Omega) &= -i\Tr{[H_d+\Omega(t)H_c, \rho(t)]\sigma_k} \\
&\quad \qquad + \Tr{\left(\Gamma \mathcal{D}[c_d] + M \mathcal{D}[c_m]\right)\rho(t) \sigma_k}, \\
\mathbf{g}_k(\mathbf{x}) &= \sqrt{\eta M}\Tr{ \mathcal{H}[c_m]\rho(t) \sigma_k}.
\end{align*}
with $\rho(t)$ substituted by \eqref{rcv}.
By accordingly, rewriting the output equation \eqref{J}, one attains 
\begin{equation}\label{Output}
dy(t)= h(\mathbf{x}(t))dt + dV(t), 
\end{equation}
where the term $dV(t) = dW(t) + dZ(t)$ includes
an additional infinitesimal Wiener increment with $\mathbb{E}[dZ(t)]=0$, $\mathbb{E}[dZ^2(t)]=\sigma_z^2 dt$ 
corresponding to the measurement noise associated with the spectrum analyzer, and the term $h(\mathbf{x}(t))$ is computed as
\[
h(\mathbf{x}) = \frac{\sqrt{\eta M}}{2}\Big( \Tr{c_m  + c_m^\dagger} + \sum_{k=1}^3 x_k(t) \Tr{ (c_m  + c_m^\dagger) \sigma_k } \Big)\cdot
\]

 \subsection{Quantum Filter Estimator}
The Quantum Filter Estimator is a fundamental tool in the estimation of quantum states. It provides a systematic approach to updating the state of a quantum system based on continuous measurement results. This estimator relies on the principles of quantum measurement theory and stochastic calculus to produce an estimate of the quantum state in real-time.

The quantum stochastic dynamics \eqref{smeftp} is inherently a filtering equation, \cite{bouten2007introduction}. Following this approach, but for the coherence vector dynamics in \eqref{dx} results that the estimated state \(\hat{\mathbf{x}}_{QF}(t)\) can be predicted through
\begin{equation}\label{xq}
d\hat{\mathbf{x}}_{QF}(t) = \mathbf{f}( \hat{\mathbf{x}}_{QF}(t),\Omega(t)) \, dt + \mathbf{g}(\hat{\mathbf{x}}_{QF}(t)) \, (dy(t)-h({\hat{\mathbf{x}}_{QF}})dt).
\end{equation}
Here, \(\hat{\mathbf{x}}_{QF}(t)\) represents the estimated state vector of quantum filtering equation at time \(t\), \(\mathbf{f}(t, \hat{\mathbf{x}}_{QF}(t))\) is the drift term, and \(\mathbf{g}(\hat{\mathbf{x}}_{QF}(t))\) is the diffusion term influenced by the Wiener increment \(dW(t)\), that in \eqref{xq} is replaced by the innovation term $dy(t)-h({\hat{\mathbf{x}}_{QF}})dt$.

\subsection{Quantum Kalman Filter estimator}
The Quantum Kalman Filter Estimator extends the classical Kalman filter to the quantum domain. It is particularly useful for systems where the noise and disturbances can be modeled as Gaussian processes. By leveraging the properties of the Kalman filter, this estimator offers a robust and efficient means of tracking the evolution of quantum states under noisy conditions, yielding high-quality estimates in the presence of measurement and process noise while also providing the associated confidence or uncertainty levels.
It is however important to stress that in this setup one have first to deal with following points: (i) the stochastic
system dynamics for the coherence vector in \eqref{dx} is nonlinear; (ii) the noise in the output equation \eqref{Output} is correlated with process noise in \eqref{dx} by noting that $dV$ includes $dW$; and (iii) the coherence vector space for a qubit is a ball with radius 1, known as the Bloch ball, thereby satisfying $\|\mathbf{x}\|^2 \leq 1$.
To address (i), we apply an Extended Kalman Filter (EKF) based approach by linearizing the non-linear functions around the current estimate using a first-order Taylor expansion.
For (ii), we  make use of a de-correlating approach, 
see \cite{simon2006optimal}, but tailored for our case, such that the new transformed equivalent system satisfies the property of uncorrelated process and measurement noises, and therefore one can then make use of the standard EKF equations.
To address (iii), we integrate the structural properties of a valid quantum state into the observer equations by constraining the state estimate $\hat{\mathbf{x}}_{KF}(t)$ at the update equation to always belong to the Bloch ball.
To this end, we define the following operator acting on the coherence vector $\mathbf{x}$ and a real vector $\mathbf{a} \in \mathbb{R}^{3}$ as
\begin{equation}\label{oplus}
\mathbf{x} \oplus \mathbf{a}= {\Proj\left( {\mathbf{x}}+{\mathbf{a}} \right)},
\end{equation}
where the operator $\Proj(\mathbf{x})= 1/\max\{1,\|\mathbf{x}\|\}\mathbf{x}$ projects $\mathbf{x}$ onto the closed ball. The operator \eqref{oplus} is pivotal within our proposed filtering algorithm, ensuring that the state and consequently the estimate of the density matrix adheres to the necessary properties of a density operator, and maintains the physical relevance of our estimated states.

With the established setup, we proceed to construct a quantum Kalman filter estimator. In this framework, we treat the measurement noise $dW(t)$ and the noise of spectrum analyzer $dZ(t)$ uncorrelated, that is $\mathbb{E}[dW(t)\,dZ(t)] = 0$.
Then, its correlation with process noise is given by $\mathbb{E}[dW(t)\, dV(t)] = dt$ since we have been considered $\mathbb{E}[dW^2(t)]=dt$. Note also that the variance $\mathbb{E}[dV^2(t)]$ is defined as $\mathbb{E}[dV^2(t)] =\sigma_{v}^2dt$, with $\sigma_{v}^2= (1 + \sigma_z^2)$. Now, we establish the following result.

\begin{prop}\label{prop1}
Consider the SDE
\begin{equation}\label{eq:bar_x}
 d\bar {\mathbf{x}}=\bar{\mathbf{f}}(\bar{\mathbf{x}}, \Omega, dy)dt
+ \mathbf{g}(\bar{\mathbf{x}})d\bar W(t),
\end{equation}
where 
$$
\bar{\mathbf{f}}(\bar {\mathbf{x}}, \Omega, dy)dt= \mathbf{f} (\bar {\mathbf{x}}, \Omega)dt  - \bar\sigma \mathbf{g}(\bar {\mathbf{x}}) h(\bar {\mathbf{x}})dt+ \bar\sigma \mathbf{g}(\bar {\mathbf{x}}) dy,
$$
$\bar\sigma = 1/(1 + \sigma_z^2)$,
and
$d\bar W = dW - \bar\sigma dV$. Then, the process and measurement noises associated with the system formed by \eqref{eq:bar_x} and \eqref{Output}  are uncorrelated, that is, $\mathbb{E}[d\bar W dV] =0$. Moreover, \eqref{eq:bar_x} is equivalent to \eqref{dx} in the sense that if both SDEs start with the same initial condition and all the exogenous signals have the same trajectory, then $\mathbf{x}(t) = \bar {\mathbf{x}}(t)$, for every $t\ge t_0$.
\end{prop}
\begin{proof}
To show uncorrelated noises, we compute 
$$
\mathbb{E}[d\bar W(t) dV(t)] = \mathbb{E}[dW(t)dV(t)]-\bar\sigma\mathbb{E}[dV^2(t)] = 0.
$$
To show the equivalence of both SDEs, we first start to note using \eqref{Output} that the term
$\bar\sigma \mathbf{g}(\mathbf{x}) \big(dy- h(\mathbf{x})dt - dV\big)$ is zero. If we now add this null term to the right-hand-side of \eqref{dx}, and replace the variable $\mathbf{x}$ by $\bar{\mathbf{x}}$, we obtain precisely \eqref{eq:bar_x}.
\end{proof}

We can now propose our state estimator that continuously updates its state estimate $\hat{\mathbf{x}}_{KF}(t)$ as new measurements are acquired
through
the propagation and update steps for the state and covariance matrix. In the propagation step 
the current estimate of
the state is propagated through the deterministic system
dynamics
\begin{equation}\label{Predictx}
    d\hat{\mathbf{x}}^-_{KF}(t) = \mathbf{\bar f}( \hat{\mathbf{x}}_{KF}(t),\Omega(t),y(t)) \, dt, 
\end{equation}
producing an a priori state estimate $\hat{\mathbf{x}}^-_{KF}(t)$. Similarly, an a priori estimation of the prediction error
covariance matrix $P_{3\times 3}$ is computed as
\begin{equation}\label{Predictp}
    d{P}^-\left( t \right)=(A(t)P(t)+P(t){{A}^{\dagger}}(t)+G(t)\sigma_{\bar w}^2(t){{G}^{\dagger}}(t))dt
\end{equation}
for which the initial guess $P(0) = P_0$ 
encodes the confidence in the estimate of the initial condition
$\hat{\mathbf{x}}_{KF}(0)$, and $\sigma_{\bar w}^2 = 1+\bar\sigma$.

Note that in \eqref{Predictp}, to evaluate the propagation of $P$, we need to compute the matrices $A(t) =\frac{\partial \bar {\mathbf{f}}(\mathbf{x}, \Omega)}{\partial \mathbf{x}}$ and $G(t)=\frac{\partial \mathbf{g}(\mathbf{x})}{\partial \mathbf{x}}$.
It turns out that only the term $\mathbf{g}(\mathbf{x})$ introduces a nonlinear term (in fact quadratic) in $\mathbf{x}$. All the other ones are linear with $\mathbf{x}$, and therefore the linearizations for these terms are not approximations. This is not surprising since the Lindblad equation can be transformed into an inhomogeneous linear vector equation using the coherence vector formulation, \cite{alicki2007n}.

Next, the update step integrates the most recent measurement by updating the prediction equations \eqref{Predictx} and \eqref{Predictp} to obtain the a posteriori estimates 
\begin{equation*}
    \begin{aligned}
         \hat{\mathbf{x}}_{KF}(t)&= \hat{\mathbf{x}}_{KF}^-(t)\oplus  K(t)\left(dy\left( t \right)-h(\mathbf{x}(t))dt
         \right),\\
 P(t)&=P^-(t)-P^-(t)C^\dag\sigma_v^{-2}
  CP^-(t).
    \end{aligned}
\end{equation*}
Here, the Kalman gain is computed as
$K(t) = P(t)C^{\dagger}\sigma_v^{-2}$,   
where $C= \big[c_1, c_2, c_3\big]$, with 
$c_k= \frac{\sqrt{\eta M}}{2}\Tr{\!(c_m  + c_m^\dagger) \sigma_k\!}$, $k=1,2,3$.

By estimating the components of the Bloch vector, we can then reconstruct an estimate for the density matrix as follows:
\begin{equation}\label{est_rcv}
\hat\rho(t) = \frac{1}{2} I + \frac{1}{2} \sum_{k=1}^{3} \hat x_{k KF}(t) \sigma_k.
\end{equation}

We can now state the following result.
\begin{prop}
Consider the proposed Kalman-based filter 
with initial condition $\hat\rho_0\in\mathbb{H}_\rho$, which implies an adequate initial condition $\hat{\mathbf{x}}_{KF}(0)$. 
Then, the state estimate $\hat\rho(t)$ 
leaves in $\mathbb{H}_\rho$ 
for every $t\ge0$. 
\end{prop}
\begin{proof}
In the propagation step, the coherence vector $\hat{\mathbf{x}}_{KF}(t)$ is the solution of the  differential equation \eqref{Predictx}, which from Proposition \ref{prop1} is equivalent to the deterministic part of \eqref{dx}. 
In the measurement step, we make use of the operator \eqref{oplus}, which guarantees $\|\hat{\mathbf{x}}_{KF}(t)\|^2 \le 1$, which implies 
positive semi-definiteness of $\hat\rho(t)$; and norm preservation since $\Tr{\frac{I}{2}+ \sum_{k} \hat{\mathbf{x}}_{k KF}(t) \sigma_k}= 1$ using the fact that $\Tr(\sigma_k)=0$.
\end{proof}

\subsection{Multiple-model adaptive  estimator}
This section presents a class of Multiple Model Adaptive Estimators (MMAEs) designed to address the challenges of quantum systems with parametric model varying dynamics or uncertainties, more precisely, in the quantum-mechanical Hamiltonian. The MMAE operates with a finite number \( N \) of selected models from the original (potentially infinite) set of system models and consists of two main components: i) a dynamic generator for \( N \) weighting signals $p_\ell(t)$, $\ell=1, \ldots, N$, and ii) a bank of \( N \) observers, each designed based on one of the selected models. In our case, the observers are precisely either the Quantum estimator or the Extended Kalman filter presented in previous subsections designed to a specific value of $H_d$.
The state estimate is produced by a weighted sum of the local state estimates $\mathbf{\hat x}_\ell(t)$ generated by the observer bank, i.e., 
\begin{equation*}
\mathbf{\hat x}(t) = \sum_{\ell=1}^N p_\ell(t) \mathbf{\hat x}_\ell(t).
\end{equation*}
By employing a set of parallel filters, each corresponding to a different model of the system's dynamics, the MMAE adapts to the most likely model based on measurement data. This provides a flexible and accurate estimation method capable of handling sudden changes in the system's behavior or parameters.

The evolution of each dynamic weight $p_\ell(t)$ is set to be piecewise constant, updating its value at instants of time $t = t_k$, according to following equation, \cite{hassani2009multiple}: 
\begin{equation} \label{eq:p}
p_\ell(t_{i+1}) = \frac{\beta_\ell e^{-w_\ell(t_k)}}{\sum_{j=1}^N p_j(t_k) \beta_j e^{-w_j(t_k)}} p_\ell(t_k), 
\end{equation}
In \eqref{eq:p}, \( \beta_\ell \) is a positive weighting constant value and \( w_\ell(t_k) \) is a scalar continuous function called an error measuring function that maps the measurable signals of the quantum system and the states of the \(\ell^{th}\) local observer to a nonnegative real value. This error measuring function evaluates the discrepancy between the observed measurements and the predicted state from the \(\ell^{th}\) model. An example is \( w_\ell(t_k) = \|y_\ell(t_k)-h(\mathbf{x}_\ell(t_k)\|^2 \).
It can be proved (using similar arguments as in \cite{hassani2009multiple}) that due to the particular structure of equation \eqref{eq:p}, the weights $p_\ell(t)$ are
positive, bounded, and the overall sum $\sum_{\ell=1}^N p_\ell(t)$ is always one for all $t\ge 0$, independently of the error measuring function $w_\ell(t)$.
In this paper, we implemented the MMAE method for both Quantum Filter (QF) and Quantum Kalman Filter (KF). The obtained results are presented in Section \ref{sec:Sim}.

\section{Lyapunov-based state feedback control}
Drawing inspiration from the concepts presented in \cite{d2021introduction}, we begin by considering a state $\rho_f\in \chi_f$, and 
select a matrix $\Pi =I-\rho_f$.
This selection leads us to define the Lyapunov function as
\begin{equation}\label{Lyapunov}
V(\rho)=\Tr{\Pi\rho}.
\end{equation}
It is important to note that $\Pi$ is Hermitian, ensuring that it possesses real eigenvalues. Furthermore, $\Pi$ is positive semi-definite, which guarantees that the Lyapunov function $V(\rho)$ is non-negative for all states $\rho$, i.e., $V({\rho}(t)) \geq 0$. 
Additionally, a significant characteristic of $\Pi$ is its commutativity with the final state $\rho_f$, that is, $[\rho_f,\Pi]=0$. 
\begin{remark}\label{stationary}
Stationary points of the Lyapunov function $V(\rho)$ are identified as states $\bar{\rho}$ that commute with $\rho_f$, see \cite{d2021introduction}, i.e., $[\bar{\rho}, \rho_f] = 0$. Given the definition of $\Pi = I - \rho_f$, this commutation property implies $[\bar{\rho}, \Pi] = 0$. 
\end{remark}

We initially derive a controller to minimize $V(\rho)$, assuming access to the actual state $\rho(t)$. Subsequently, we adapt this approach for situations lacking direct access to $\rho(t)$, employing an estimated state $\hat\rho(t)$.

\begin{thm}
Consider the quantum system \eqref{smeftp} and let $\chi_f=\{\rho\in\mathbb{H}: [\rho,\Pi]=0\}$ be the desired stationary set. Consider the feedback control 
\begin{equation}\label{controller}
\Omega(t) =
\begin{cases}
-i\frac{\Upsilon_0}{\Upsilon_1} & \text{   for } t \in [t_{k-1}, t_k),\ k = 1, 3, 5, \ldots \\
\quad 0 & \text{    otherwise}
\end{cases}
\end{equation}
where ${{\Upsilon }_{0}}=\Gamma \Tr{\Pi \mathcal{D}\left[ {c_d} \right]{{\rho }}} 
+ M \Tr{\Pi \mathcal{D}\left[ {c_m} \right]{{\rho }}}+\alpha \Tr{\Pi{{\rho }}}$, with $\alpha\in \mathbb{R}^+$, and ${{\Upsilon }_{1}}=\Tr{\Pi \left[ {{{H_c}}},\rho \right]}$. The control value switches at event times $\!\{t_0, t_1, \ldots \!\}$ when $\Upsilon$ crosses an $\varepsilon>0$ value and the period of time between the last switching and current is larger than some dwell time, that is,  
$t_0$ the initial time and $t_k$ is defined as
\begin{equation*}
\begin{aligned}
\left\{
\begin{array}{ll}
t_k = \min\limits_t \{ t \ge t_{k-1} + \Delta_1 : \Upsilon_{1} \le \varepsilon \} & \text{for } k \text{ odd}, \\
t_k = \min\limits_t \{ t \ge t_{k-1} + \Delta_2 : \Upsilon_{1} \ge \varepsilon \} & \text{for } k \text{ even},
\end{array}
\right.
\end{aligned}
\end{equation*}
for some appropriate dwell times
$\Delta_1,\, \Delta_2\in \mathbb{R}^+$. 
Then, the resulting closed-loop system exhibits the property that the set $\chi_f$ is practically exponentially stable in probability.
\end{thm}
\begin{proof}
Given that $\rho_f$ is selected to commute with $\Pi$, it follows that $\rho_f$ becomes a stationary point of the Lyapunov function $V(\rho)$, as discussed in \cite{d2021introduction}. Furthermore, according to Remark \ref{stationary}, the set $\chi_f$ is identified as a stationary set. To investigate its stability, we consider the infinitesimal generator of $\rho(t)$ applied to $V(\rho)$. 
This analysis involves the diffusion process described in \eqref{smeftp} and leverages the Itô formula, resulting in
\begin{equation*}
\!L(V(\rho))\!=\!\Tr{\Pi^T \! (\Delta(\rho) \!+\! \Epsilon(\rho)\Omega(t))}\!+\!\frac{1}{2}\Tr{Z(\rho)H_V Z^T(\rho)}
\end{equation*}
where $L(V(\rho))$ is the quantum Markov generator for the Lyapunov function $V$, 
$\Delta(\rho)=-i\left[{H_d,\rho \left( t \right)} \right] + \Gamma\mathcal{D}\left[ {{c_d}} \right]\rho(t)+M\mathcal{D}\left[ {{c_m}} \right]\rho(t)$, $E(\rho)=-i\left[ {H_c,\rho \left( t \right)} \right]$, $Z(\rho)= \sqrt {\eta M}\mathcal{H}[c_m]\rho(t)$.
The Hessian matrix $H_V$ of $V(\rho)$ is zero in this context. Considering the commutation relation $[\rho_f, {H}_d] = 0$ and the definition $\Pi = I - \rho_f$, we deduce $[\Pi, H_d] = 0$. By employing the trace property $\Tr{A[B,C]} = \Tr{C[A,B]}$, it follows that $\Tr{\Pi[H_d, \rho]}= 0$. Thus, we obtain
\begin{align}\label{LVD}
L(V(\rho)) &= -i\Omega(t)\Tr(\Pi[{H}_c, \rho(t)]) + \Gamma \Tr(\Pi\mathcal{D}[c_d]\rho(t)) \notag \\
&\quad+M \Tr(\Pi\mathcal{D}[c_m]\rho(t)).   
\end{align}
At this point, our objective is to identify a control function $\Omega(t)$ that satisfies $L(V(\rho)) \le -\alpha V(\rho)$, ensuring that $V(t) = V(\rho(t))$ is a supermartingale and 
decreases exponentially with a rate $\alpha > 0$. 
Among the potential solutions, we opt for the one that minimizes the control effort $\Omega^2(t)$ pointwise in time. Hence, we seek $\Omega(t)$ that solves the optimization problem:
\begin{equation*}
\begin{aligned}
     & \underset{\Omega(t)}{\text{minimize}} && \Omega^2(t)\\
     & \text{subject to} && -i\Omega(t) \Upsilon_1 + \Upsilon_0 \le 0.
 \end{aligned}   
 \end{equation*}
The solution to this optimization problem is precisely the control law presented in \eqref{controller}. Caution is advised when $\Upsilon_1$ approaches zero to prevent computational inaccuracies and excessive control signal magnitudes. As a preventive measure, the control strategy switches $\Omega(t)$ to zero once $\Upsilon_1 \le \varepsilon$. To avoid Zeno behavior, characterized by an infinite number of switches within a finite time, we enforce a minimum dwell time between switching, $t_k - t_{k-1} > \Delta$.

Taking now into account the results and reasoning presented in \cite{aguiar2007switched} for switching systems with stable and unstable modes, and acknowledging that $L(V)$ remains bounded when $\Omega = 0$, the practical stability in probability of the set $\chi$ is conclusively established by noticing that $L(V)\le -\eta(V) + \eta_0$, for some class $\mathcal{K}$ function $\eta$ and constant $\eta_0\ge 0$.
\end{proof}

\begin{figure*}[t]
	\centering
	\includegraphics[width=\textwidth]{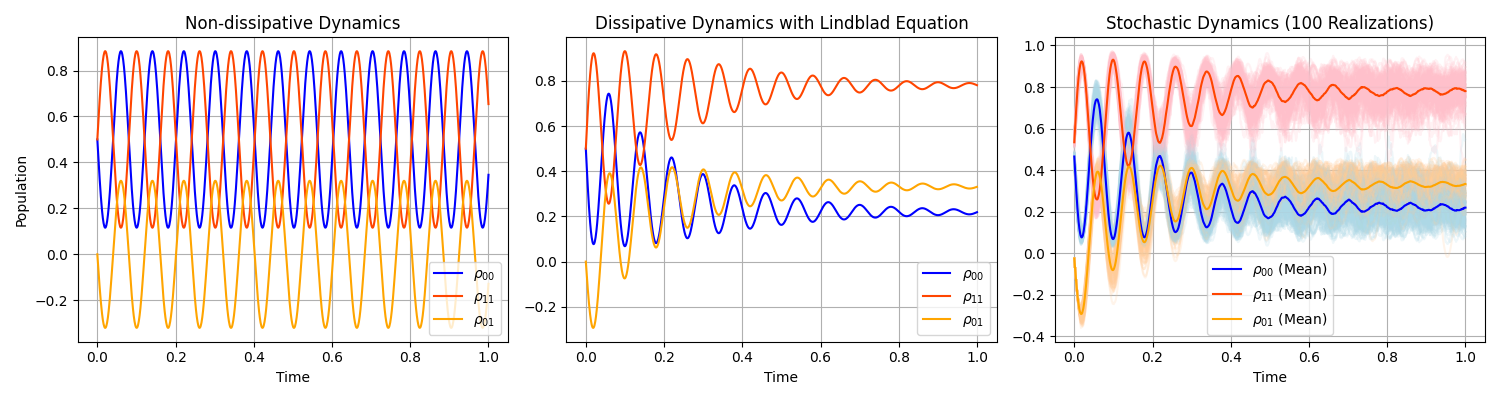}
        \caption{\small Comparative evolution of quantum state populations. This figure illustrates the population dynamics of a quantum system under three different scenarios: (a) Non-dissipative dynamics, (b) Dissipative dynamics with the Lindblad equation, and (c) Stochastic dynamics over 100 realizations. In each subplot, the populations $\rho_{00}(t)$, $\rho_{11}(t)$, and the real part of the off-diagonal $\rho_{01}(t)$ (or $\rho_{10}(t)$) are plotted in time. The non-dissipative dynamics exhibit periodic oscillations, the dissipative dynamics show damped oscillations leading to steady-state populations, and the stochastic dynamics reveal ensemble-averaged behaviors with individual realizations contributing to the shaded regions.}
        \label{dynamics}
\end{figure*}

\begin{thm}
Consider the same conditions of Theorem 1, but now replacing \eqref{controller} by  
\begin{equation}\label{controller1}
\Omega \left( t \right)=
   -i\displaystyle {\frac{{{\hat\Upsilon }_{0}}}{{{\hat\Upsilon }_{1}}}},\quad  t \in [t_{k-1}, t_k),\  k=1, 3, 5,\ldots  
\end{equation}
where ${\hat{\Upsilon }_{0}}=\Gamma \Tr{\Pi \mathcal{D}\left[ {c_d} \right]{\hat{\rho }}}
+ M \Tr{\Pi \mathcal{D}\left[ {c_m} \right]{\hat{\rho }}}+\alpha \Tr{\Pi{\hat{\rho }}}$, with $\alpha\in \mathbb{R}^+$, and ${\hat{\Upsilon }_{1}}=\Tr{\Pi \left[ {{{H_c}}},\hat{\rho} \right]}$ with $\hat\rho$ denoting the output estimate of the quantum observer. Let $\tilde\sigma^2(t)$
be the incremental variance associated to the noisy estimation error $\Tr{\tilde{\rho}}=\Tr{\rho-\hat{\rho}}$, which is considered to be bounded. Then, the set $\chi_f$
is NSpS with respect to $\tilde\sigma^2(t)$.
\end{thm}

\begin{proof}
From \eqref{LVD}, it follows that $L(V(\rho))$ satisfies 
\[
L(V(\rho ))=-i{{\hat{\Upsilon }}_{1}}\Omega \left( t \right)+{{\hat{\Upsilon }}_{0}}-i{{\tilde{\Upsilon }}_{1}}\Omega \left( t \right)+{{\tilde{\Upsilon }}_{0}}
\]
where $\tilde{\Upsilon}_0={\Upsilon}_0-\hat{\Upsilon}_0$ and $\tilde{\Upsilon}_1={\Upsilon}_1-\hat{\Upsilon}_1$.
Since $\Omega$ is bounded, $\tilde{\Upsilon}_0$ and $\tilde{\Upsilon}_1$ are functions that can be bounded by class $\mathcal{K}$ functions of $\tilde\sigma^2$, it follows that using the same arguments in the previous proof and  \cite{aguiar2007switched,deng2001stabilization} one can conclude that for $V(\rho)\ge \eta_1(\tilde\sigma^2)$, the generator 
$L(V)\le -\eta(V) + \eta_0$, for some class $\mathcal{K}$ function $\eta$, positive definite function $\eta_1$ and constant $\eta_0\ge 0$. 
Therefore, the results follow.
\end{proof}

\section{Application to Qubits in Leaky Cavities 
} \label{sec:Sim}
The methodology presented here for the illustrated example is universally applicable to any two-level atom experiencing classical influences such as driving, damping, and detuning. Our examination specifically focuses on a qubit interacting with a leaky cavity 
in 
a rotating frame. 
The drift Hamiltonian, $H_{d}=\frac{\omega_R}{2}{\sigma}_3$, incorporates the effective detuning, $\omega_\Delta=\omega_R+\Delta \omega_R-\omega_0$, which represents the adjusted frequency discrepancy of the qubit. Here, $\omega_R$ denotes the frequency gap between the excited $\left| e \right\rangle$ and ground $\left| g \right\rangle$ states, $\omega_0$ signifies the oscillation frequency of the driving field, and $\Delta \omega_R$ refers to the frequency shift. The Pauli operator ${\sigma}_3=\left| g \right\rangle \left\langle g \right|-\left| e \right\rangle \left\langle e \right|$ serves as the inversion operator for the qubit. The control Hamiltonian, $H_c=-{\Omega(t)}{\sigma}_1$, where $\Omega(t)$ is the Rabi frequency, is directly proportional to the transition dipole moment and the amplitude of the driving electromagnetic field. For the annihilation and creation field operators, we utilize the Pauli lowering and raising operators, ${a}={\sigma}_-$ and ${a}^\dagger={\sigma}_+$, respectively. The initial density operator is chosen as $\rho_0 =\small \begin{pmatrix} 0.5 & -0.5i \\ 0.5i & 0.5 \end{pmatrix}$ corresponding to the Bloch vector ${{\mathbf{x}}}_0=[0,1,0]^T$.
The parameters used in simulations are defined as follows. The decay rate (\(\Gamma\)) is set to 10 \(s^{-1}\), representing the rate at which the system loses energy to its environment. The Rabi frequency (\(\omega_R\)) is \(5 \Gamma\) rad/s, characterizing the frequency of oscillation for a two-level system driven by a resonant electromagnetic field. The interaction strength (\(M\)) used in the Lindblad term is set to \(1 \, s^{-1}\), representing the strength of certain dissipative processes. The detection efficiency (\(\eta\)) is set to 0.8, indicating that 80\% of the signals are correctly detected. The time step for numerical integration (\(dt\)) is set to 0.001 seconds to improve accuracy, and the total simulation time (\(T\)) is set to 1.0 second, defining the duration over which the simulation is run. For now, we consider the amplitude of the driving field (\(\Omega\)), which represents the strength of the external field, as \(3 \Gamma\) rad/s. We later change this value with the one obtained from control method. 

Our analysis commences with a comparison across three distinct scenarios, as depicted in Fig. \ref{dynamics}. These include the evolution absent any dissipation or noise, showcasing the ideal oscillatory behavior due to coherent superposition of states; the dissipative dynamics governed by the Lindblad Equation, which simulates the environmental impact on the quantum states leading to oscillation decay; and the stochastic dynamics that introduce stochastic noise to represent random environmental disturbances. 
Notably, all stochastic evolution graphs are plotted across 100 realizations alongside their mean values.

\begin{figure*}[t]
    \centering
    \includegraphics[width=\textwidth]{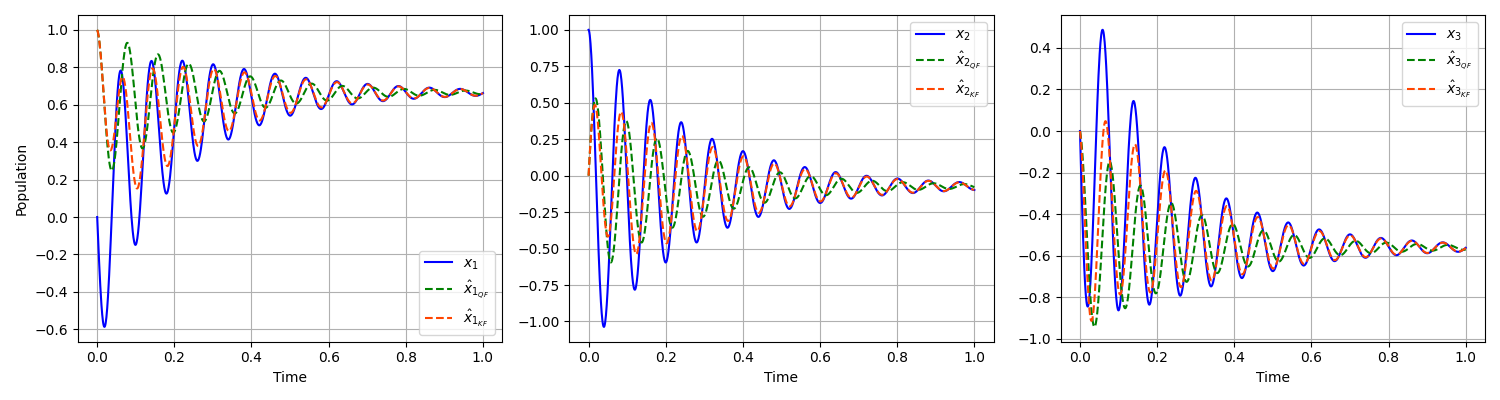} 
    \caption{\small Comparative analysis of true state, Kalman filter, and quantum filter in the \(\mathbf{x}\)-representation. The plots show the population dynamics over time for three different state variables in the coherence vector representation: (a) \(x_1\), (b) \(x_2\), and (c) \(x_3\). The blue solid lines represent the true state (\(x_1, x_2, x_3\)), the green dashed lines represent the estimates obtained from the quantum filter (\(\hat{x}_{1QF}, \hat{x}_{2QF}, \hat{x}_{3QF}\)), and the red dashed lines represent the estimates from the Kalman filter (\(\hat{x}_{1KF}, \hat{x}_{2KF}, \hat{x}_{3KF}\)).}
    \label{xQFKF}
\end{figure*}
\begin{figure*}[t]
    \centering
    \includegraphics[width=\textwidth]{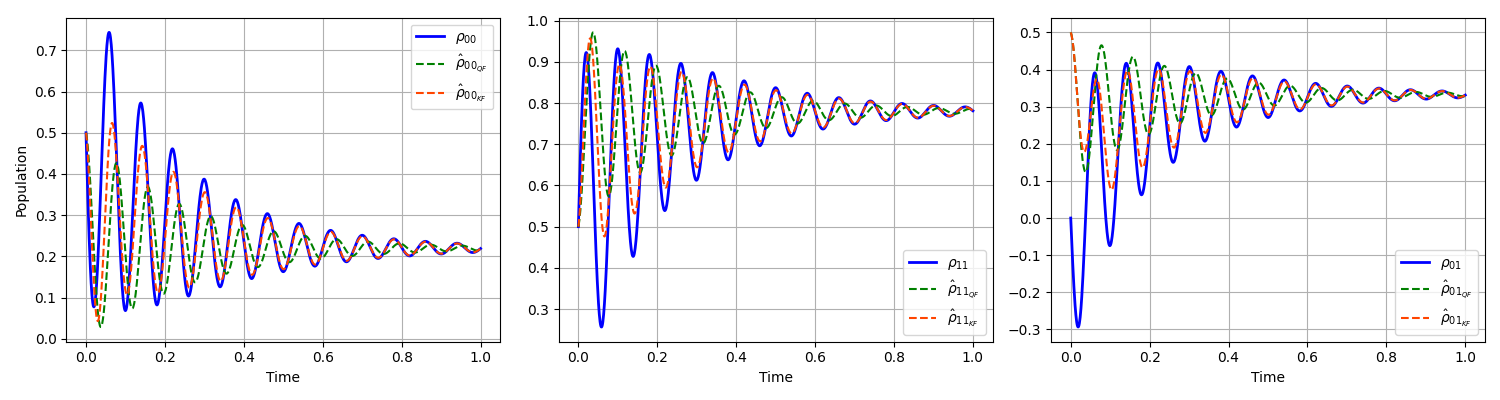} 
    \caption{\small Comparison of true density matrix elements with quantum and Kalman filter estimates. The plots display the evolution of the real parts of the density matrix elements over time, comparing the true values (solid blue lines) with the estimates obtained from the quantum filter (dashed green lines) and the Kalman filter (dashed red lines).}
    \label{RHOfromx}
\end{figure*}

We now consider the quantum observer. First, we illustrate a comparative analysis of quantum filter estimator and quantum Kalman filter estimator for the elements of coherence vectors as depicted in Fig. \ref{xQFKF}. For both filters the initial value is ${\hat{\mathbf{x}} }_0= [1, 0, 0]^T$, corresponding to $\hat\rho_0 =\small \begin{pmatrix} 0.5 & 0.5 \\ 0.5 & 0.5 \end{pmatrix}$.
Analysis shows that both filters closely track the true state across all three variables, with minor deviations occurring primarily during the initial transient phase. The Kalman filter estimates (red dashed lines) show a faster convergence to the true state compared to the quantum filter estimates (green dashed lines), especially evident in the initial peaks. However, both filters ultimately achieve almost similar accuracy in steady-state estimation. The slight deviations observed in the quantum filter during the early phase suggest sensitivity to initial conditions and noise, which is corrected over time. Overall, the effectiveness of both filters in state estimation is demonstrated, with each providing valuable insights into the dynamics of the quantum system.

We now reconstruct density matrix $\hat{\rho}(t)$ from the elements of coherence vector \(\hat{\mathbf{x}}\). This process involves transforming the 
state space representation back into the density matrix form, ensuring that the estimates are consistent with the physical properties of the quantum state. As illustrated in Fig. \ref{xQFKF}, both the quantum filter and Kalman filter estimates converge to the true values for all three elements. This demonstrates the effectiveness of both filters in accurately estimating the state of the quantum system, however, the Kalman filter appears to have a faster convergence compared to the quantum filter, particularly evident in the $\rho_{00}$ and $\rho_{01}$. In the steady-state region, both filters provide estimates that closely match the true values, indicating that both methods are robust for long-term state estimation. Overall, the minor discrepancies observed during the transient phase are corrected over time, and both filters show high fidelity in tracking the true state of the system. This analysis highlights the robustness and accuracy of the quantum and Kalman filters in estimating the state of a quantum system.

We also show the fidelity between the true density operator and the estimated state computed through $ F(\rho,\hat{\rho}) = \left( \Tr{ \sqrt{ \sqrt{\rho} \hat{\rho} \sqrt{\rho} } }  \right)^2$ for both quantum filter and Kalman filter. This measure quantifies how close the estimated density matrix is to the true density matrix, with values ranging from 0 (completely different) to 1 (identical). As illustrated in Fig. \ref{fidelity}, both filters show high fidelity with the true state, particularly in the steady state, indicating accurate tracking of the quantum state. Initially, there are slight discrepancies, but both filters converge to a high fidelity value, demonstrating their effectiveness in quantum state estimation. The transient region shows fluctuations, reflecting the filters' adaptation to the initial state dynamics. The Kalman filter exhibits slightly faster convergence compared to the quantum filter, which is crucial for real-time applications. This figure underscores the robustness and precision of both the Kalman and quantum filters in maintaining high fidelity with the true quantum state over time.
\begin{figure}[t]
    \centering
    \includegraphics[scale=0.4]{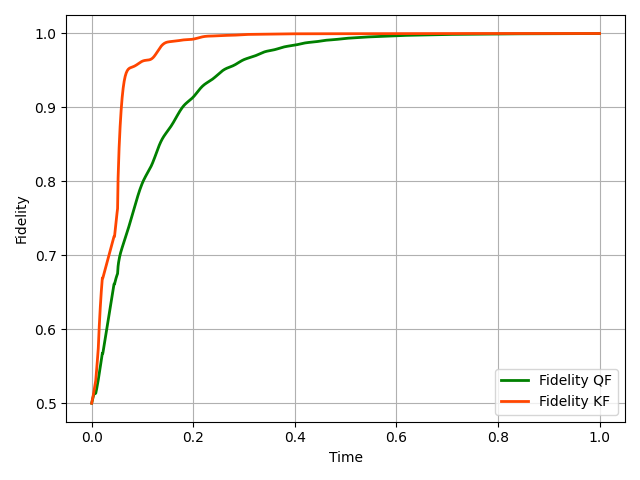} 
    \caption{\small Fidelity comparison over time between true and estimated density matrices. The plot shows the fidelity values between the true density matrix and the estimates obtained from the Kalman filter and the quantum filter over time.}
    \label{fidelity}
\end{figure}

It is important to emphasize that the Kalman filter offers the significant advantage of allowing the tuning of its gains (specifically, the covariances $Q$ and $R$) to enhance performance. In contrast, the Quantum filter lacks this flexibility. 

In order to delve more into the analysis of proposed Kalman filtering, we also show the trace of the error covariance matrix \( P \). The trace of \( P \) serves as a crucial indicator of the Kalman filter's performance, quantifying the evolution of the uncertainty in the state estimates and demonstrating the filter's effectiveness in reducing this uncertainty over time.
As depicted in Fig. \ref{traceP}, the trace of \( P \) is initially high, reflecting significant uncertainty in the initial state estimates. This high uncertainty is typical at the start of the filtering process, where the initial state is often based on a rough estimate or guess. As the Kalman filter processes measurements over time, the trace of \( P \) decreases, indicating that the filter is effectively reducing the uncertainty in the state estimates by incorporating new measurement information. This reduction in uncertainty is evident from the rapid decline in the trace of \( P \) during the initial phase of filtering. Eventually, the trace of \( P \) reaches a steady-state value, signifying that the filter has minimized the uncertainty in the state estimates. The steady-state value of the trace of \( P \) is influenced by the process noise and measurement noise. In an ideal scenario with perfect measurements, the trace of \( P \) would approach zero. However, due to the inherent noise in real systems, a non-zero steady-state value is typical, representing the residual estimation error.
\begin{figure}[t]
    \centering
    \includegraphics[scale=0.4]{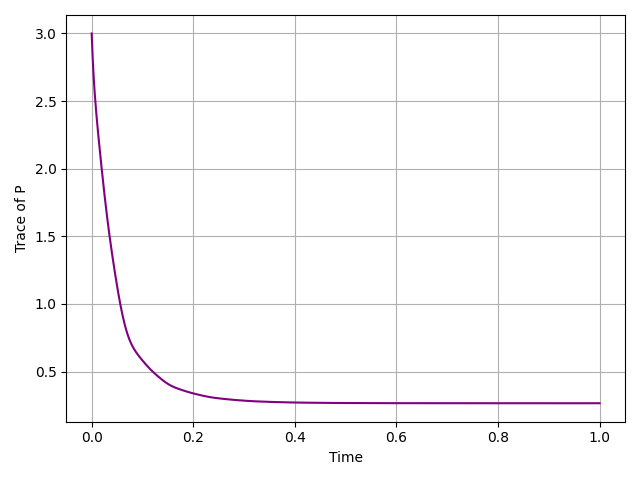} 
    \caption{\small Evolution of the Trace of the Error Covariance Matrix $P$ over Time. The plot shows the trace of the error covariance matrix $P$ as a function of time in the Kalman filtering process. The trace of $P$ is a scalar measure of the total estimation error variance across all state variables.}
    \label{traceP}
\end{figure}

Having implemented both filters, we now delve into the multiple model implementation for both quantum and Kalman filtering approaches. We evaluate the efficacy of MMAEs for quantum state estimation under uncertainties in the parameter $\omega_R$. We consider five different models with $\omega_R$ values set to 0.8, 0.9, 1.0, 1.1, and 1.2 times the nominal value. Figure \ref{multi} presents the time evolution of the quantum state populations $x_1$, $x_2$, and $x_3$, with each row corresponding to a specific $\omega_R$ multiplier and each column representing a state variable. Notably, when $\omega_R$ is exactly the nominal value (1.0), both QF and KF estimates converge to the true states, demonstrating their accuracy and robustness. However, for other $\omega_R$ values (0.8, 0.9, 1.1, 1.2), there are noticeable discrepancies between the true states and the filter estimates. These discrepancies highlight the sensitivity of the system's dynamics to variations in $\omega_R$ and underscore the importance of adaptive estimation techniques. These results underscore the robustness of the MMAE approach, where dynamic weights adjust the contributions of each model based on the observed data. The adaptive nature of the MMAE ensures that the most probable model predominantly influences the state estimation, thereby effectively handling the uncertainties in $\omega_R$. This adaptability is crucial for maintaining accurate state estimates in quantum systems with variable dynamics, providing a reliable framework for practical quantum state monitoring and control applications. 

\begin{figure}[t]
    \centering
    \includegraphics[width=0.5\textwidth]{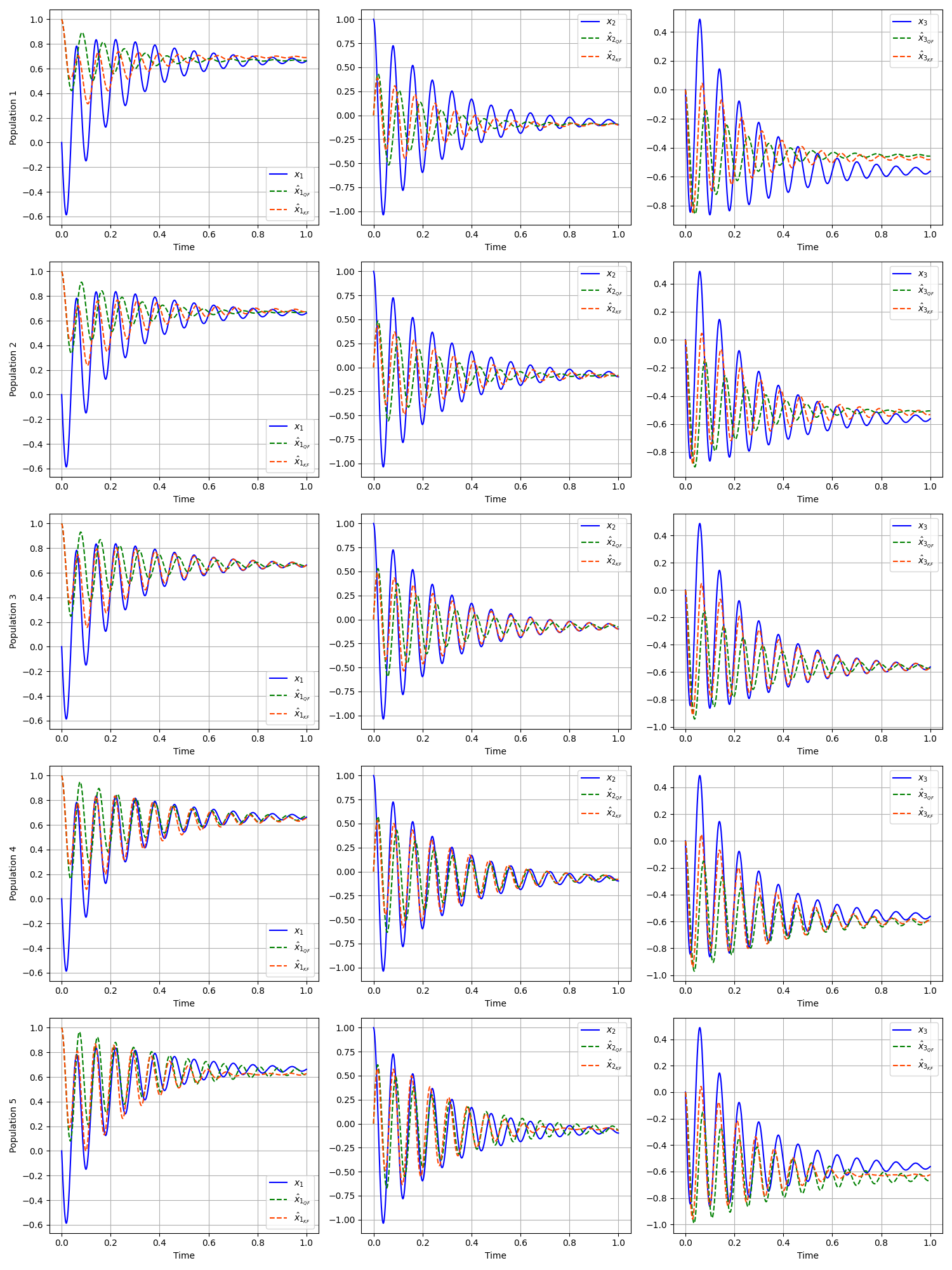} 
    \caption{\small Time Evolution of quantum state populations with multiple models of $\omega_R$. This figure illustrates the dynamic behavior of the quantum state populations $x_1$, $x_2$, and $x_3$ under varying $\omega_R$ values, with true states shown in blue, quantum filter estimates in green, and Kalman filter estimates in orange. Each row corresponds to a different $\omega_R$ multiplier: 0.8 (Population 1), 0.9 (Population 2), 1 (Population 3), 1.1 (Population 4), and 1.2 (Population 5). Columns represent the state variables $x_1$, $x_2$, and $x_3$ respectively. The plots highlight the performance of the QF and KF in tracking the true quantum state across different $\omega_R$ scenarios.}
    \label{multi}
\end{figure}

In addition, we demonstrate the dynamic adaptation of model weights in the MMAE framework. Each weight $p_i$ corresponds to a model with a different $\omega_R$ multiplier.
At the start of the simulation, all weights are initialized equally, reflecting no prior knowledge about which model is most accurate. As the estimation process progresses, the weights dynamically adjust based on the error measures between the observed measurements and the predicted states from each model. This adjustment is governed by a dynamic recursion equation, which allows the MMAE to identify and emphasize the most likely model over time. The results depicted in Fig. \ref{prob} show that the weight $p_3$, corresponding to the nominal $\omega_R$, gradually dominates as the most probable model. This is evidenced by its increase and stabilization at a higher value compared to the other weights, which decrease as the system identifies them as less probable. The convergence of $p_3$ demonstrates the effectiveness of the MMAE in correctly identifying the model that best represents the true system dynamics.
The variations and eventual stabilization of the weights illustrate the MMAE's capability to adapt to the true system parameters, providing robust state estimation even under model uncertainties. This dynamic adaptation is crucial for accurately tracking the quantum state in systems where parameters may not be precisely known or may vary over time.

\begin{figure}[t]
    \centering
    \includegraphics[scale = 0.4]{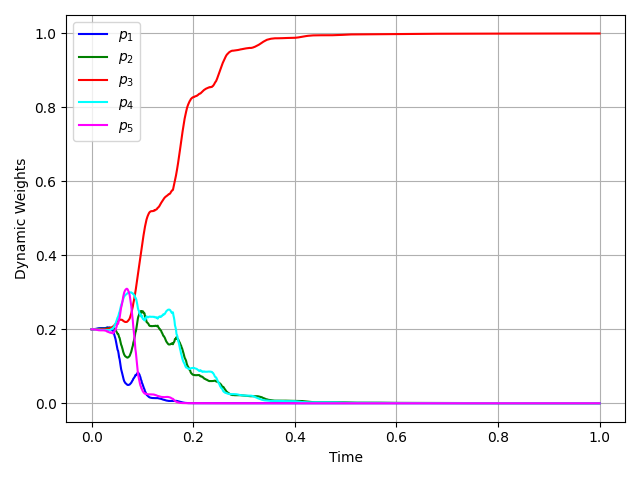} 
    \caption{\small Dynamic evolution of model weights $p_i$ in the MMAE framework. This figure illustrates the time evolution of the dynamic weights $p_1$, $p_2$, $p_3$, $p_4$, and $p_5$ associated with different models in the MMAE. The weights are represented to show their adaptation over time as the MMAE converges to the most probable model.}
    \label{prob}
\end{figure}

We now present our results regarding the dynamic control of the studied quantum system under stochastic influences, aiming to drive the system towards a specified target set. All simulations are conducted over 100 realizations to account for the effects of noise and randomness. As illustrated in Fig. \ref{control}, the first subplot displays the control signal $\Omega(t)$
over time for each of the 100 realizations, shown in light blue together with the mean control signal highlighted. This control signal is computed using feedback from the system's state, adjusted dynamically to steer the system towards the desired target state, using our Lyapunov switching approach. The variability in the control signals across realizations indicates the system's adaptive response to stochastic disturbances. The second subplot shows the evolution of the density matrix elements over time, where now we use the control signal $\Omega(t)$ to achieve the controlled evolution. This plot illustrates that the trajectory of the state populations under the control signal converge towards the desired target. The third subplot focuses on the fidelity of the system's state relative to the target, computed by $ F(\rho,\rho_f) = \left( \Tr{ \sqrt{ \sqrt{\rho_f} \rho \sqrt{\rho_f} } }  \right)^2$ for $\rho_f \in \chi_f$. This plot confirms that the control strategy effectively drives the system towards the target state, as evidenced by the high fidelity values over time.

\begin{figure*}[t]
    \centering
    \includegraphics[width=\textwidth]{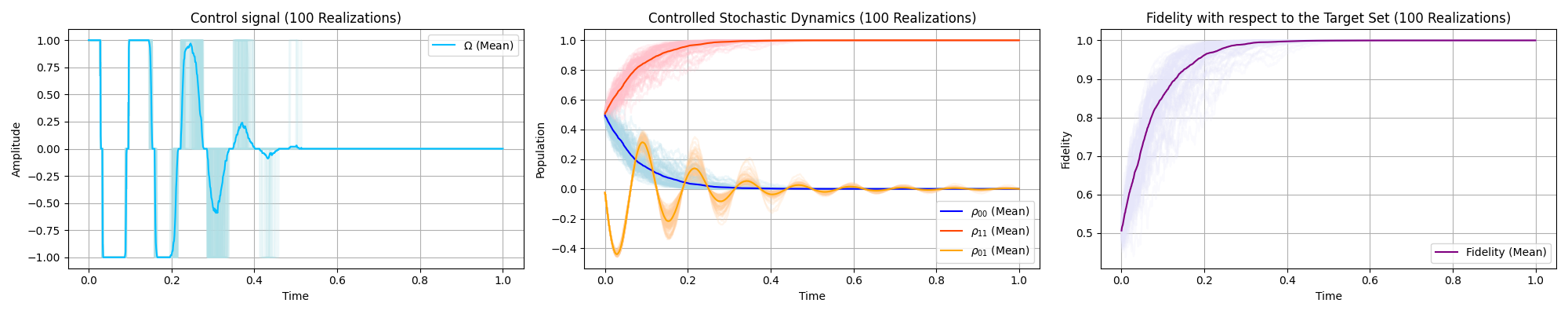} 
    \caption{\small Control signal, controlled state evolution, and fidelity over 100 stochastic realizations. The first subplot displays the control signal $\Omega(t)$ applied over time for 100 realizations (light blue lines) with the mean control signal. The second subplot illustrates the controlled time evolution of the elements of the density operator, depicting how the state evolves towards the target set. The third subplot presents the fidelity with respect to the target state, demonstrating that the target state is achieved with high accuracy.}
    \label{control}
\end{figure*}

\section{Conclusion}

This paper presents a robust framework for the control of stochastic dynamics in quantum systems using a Lyapunov-based control approach with homodyne measurement. To address the problem of estimating the density operator $\rho(t)$ from these measurements, we examined two filtering methods: traditional quantum filtering and an extended Kalman filtering variant, both aimed at estimating the coherence vector elements \( \mathbf{x}(t) \).  
Our results illustrate the superiority of the Kalman based filter in terms of performance over the traditional Quantum filter, largely due to the tunable covariances that allow for faster convergence of the estimation error. Despite the increased complexity in managing the dynamics of a stochastic master equation with correlated noise, the proposed Kalman filter effectively preserves the quantum properties of the estimated state variable $\hat \rho(t)$. We further addressed scenarios with unknown quantum-mechanical Hamiltonians and system uncertainties, where both filters showed degraded performance and increased estimation errors. To counteract this, we proposed a multiple model estimation scheme, enhancing estimation accuracy under uncertainty.
The estimated density operator \( \hat{\rho} \) was then employed in a switching-based Lyapunov control scheme. This control strategy, utilizing \( \hat{\rho} \), achieved noise-to-state practical stability in probability with respect to the estimation error variance. Our approach effectively stabilized a qubit coupled to a leaky cavity under homodyne detection, even with resonance frequency uncertainties. The integration of multiple model estimation with Lyapunov-based control offers a significant advancement in quantum control, ensuring high fidelity and stability despite noise and uncertainties. This work has potential applications in quantum information processing and metrology, with future efforts directed towards extending these methods to more complex quantum systems and real-time implementations. All of the methods for the estimation and control presented in this work can be conveniently extended to N-level systems. However, one has to be careful for extending the Kalman filter to N-level since the projection operator has to be modified such that the updated state remains on a proper subset of the corresponding ball.

\bibliographystyle{IEEEtran}
\bibliography{IEEEabrv,ref}

\end{document}